\documentclass[preprint,aps,nofootinbib,superscriptaddress]{revtex4-1}
\usepackage{amssymb,mathrsfs,physics,amsthm,amssymb,amsfonts,amsmath}

\usepackage{hyperref}
\hypersetup{
    plainpages=false,       
    unicode=false,          
    pdftoolbar=true,        
    pdfmenubar=true,        
    pdffitwindow=false,     
    pdfstartview={FitH},    
    pdfnewwindow=true,      
    colorlinks=true,        
    linkcolor=Blue,         
    citecolor=Red,        
    filecolor=magenta,      
    urlcolor=blue           
}
\usepackage[usenames, dvipsnames]{xcolor}

\newcommand\id{\leavevmode\hbox{\small1\kern-3.3pt\normalsize1}}

\usepackage{enumitem} 

\usepackage[normalem]{ulem}

\theoremstyle{definition}
\newtheorem{theorem}{Theorem}

\usepackage{bbm}

\usepackage{graphicx}

\usepackage{tikz}
\usetikzlibrary{arrows}
\usepackage{mathptmx}
\usepackage{moresize}

\usepackage{CJKutf8}

\begin{document}

\begin{abstract}
The expected indefinite causal structure in quantum gravity poses a challenge to the notion of entanglement: If two parties are in an indefinite causal relation of being causally connected and not, can they still be entangled? If so, how does one measure the amount of entanglement? We propose to generalize the notions of entanglement and entanglement measure to address these questions. 

Importantly, the generalization opens the path to study quantum entanglement of states, channels, networks and processes with definite or indefinite causal structure in a unified fashion, e.g., we show that the entanglement distillation capacity of a state, the quantum communication capacity of a  channel, and the entanglement generation capacity of a network or a process are different manifestations of one and the same entanglement measure.


\end{abstract}

\title{Generalizing Entanglement}
\begin{CJK*}{UTF8}{gbsn}
\author{Ding Jia (贾丁)}
\email{ding.jia@uwaterloo.ca}
\affiliation{Department of Applied Mathematics, University of Waterloo, Waterloo, ON, N2L 3G1, Canada}
\affiliation{Perimeter Institute for Theoretical Physics, Waterloo, ON, N2L 2Y5, Canada}

\maketitle
\end{CJK*}

\section{Introduction}


One of the most important lessons of general relativity is that the spacetime causal structure is dynamical. In quantum theory dynamical variables are subject to quantum indefiniteness. It is naturally expected that in quantum gravity causal structure reveals its quantum probabilistic aspect, and indefinite causal structure arises as a new feature of the theory \cite{hardy_probability_2005, hardy2007towards, oreshkov2012quantum, chiribella2013quantum, brukner2014quantum, hardy2016operational}. The causal relation of two parties can be indefinite regarding whether they are causally connected or disconnected. In the context of the causal structure of spacetime, it is possible that two parties are ``in a superposition'' of being spacelike and timelike. This poses a serious challenge to the notion of entanglement --- can we still meaningfully talk about entanglement for such parties?

To address this question, we propose to generalize the notions of entanglement and entanglement measure. We base our study of generalized entanglement in the framework of process matrices \cite{oreshkov2012quantum, araujo2015witnessing, oreshkov2016causal} (we briefly review the framework in Sec. \ref{sec:pf}) that is designed to incorporate indefinite causal structure. While traditionally entanglement and entanglement measures are defined in the local operations and classical communications (LOCC) paradigm, we will see that this paradigm needs some ``fine-graining'' to incorporate nontrivial causal structure of the parties (Sec.s \ref{sec:ls}). In the polished LOCC paradigm the new generalized notions of entanglement and entanglement measures can then be introduced in very natural ways (Sec. \ref{sec:pem}).

The generalizations not only render entanglement meaningful with indefinite causal structure, but also yield new perspectives on entanglement with definite causal structure. While the process framework is designed to incorporate indefinite causal structure, it includes quantum theory with definite causal structure as a subtheory. The generalization automatically assigns entanglement to quantum channels and quantum networks with definite causal structure, opening the path to study quantum correlations of states, channels, networks, and processes in a unified fashion. To illustrate this point, we show that the entanglement distillation capacity of a state, the quantum communication capacity of a channel, and the entanglement generation capacity of a network or a process are expressed using a single entanglement measure and studied in a unified way (Sec. \ref{sec:cibm}). This highlights coherent information as a prominent example in the family of entanglement measures as it comes with significant operational meaning. The generalization also suggests new perspectives in quantum gravity. We use toy models to illustrate how the new notions allow us to generalize consideration of entanglements from bifurcation of spacelike regions to that of spacetime regions (Sec. \ref{sec:ecn}). 
  


The generalization of entanglement considered in this paper therefore serves two purposes. It first sets a framework to study entanglement when indefinite causal structure is present. To our best knowledge this question has not been investigated before, except in the thesis of the same author \cite{jia2017quantum}. (See \cite{markes2011entropy} for the related question of how to define entropy with indefinite causal structure present.) The generalization, moreover, applies the concept of entanglement to each state, channel, network, and process that can be separated into subsystems. (In this paper we focus on bipartite entanglement and leave the detailed study of multipartite entanglement for future work.) Although we are not aware of any previous work that extends the notion of entanglement to this generality, there are certainly many related works that propose generalizations of entanglement beyond states and spacelike settings. Historically, early interests in entanglement arise from its relevance for nonlocality and nonclassicality \cite{einstein1935can, schrodinger1935gegenwartige, bell1964instein}, and Bell-type inequalities had been generalized from spacelike to timelike settings \cite{leggett1985quantum}. Since then there have been attempts to introduce or study timelike entanglement in different contexts (see e.g., \cite{brukner2004quantum, fitzsimons2015quantum, olson2011entanglement}). Logically speaking a natural question that arises is whether timelike and spacelike entanglement can be considered in a unified fashion. Indeed, we answer this question in the affirmative in this paper. Regarding this question, we mention the early work of \cite{horodecki2000unified}, which presents a unified view on quantum correlations within states and channels in the context of information communication. Theorem \ref{th:pegc} below further develops the unified view to include general networks and processes. A further question that arises in the unified studies of entanglement is if one can partition the systems in nonstandard ways. For instance, although it is natural to take the input and output systems of a channel as the two subsystems, there certainly exist other ways of partition. Can one make sense of entanglement for general partitions? Along this line, we mention the previous works of \cite{oppenheim2004probabilistic, hosur2016chaos}, which considered particular entanglement measures for channels with more general partitions. We show below that general partitions can also make sense for networks and processes, and this is so for general but not just particular entanglement measures.

\section{The process framework}\label{sec:pf}
In this section we briefly review the process framework. In the original articles \cite{oreshkov2012quantum, araujo2015witnessing, oreshkov2016causal} processes are represented by unnormalized operators. We adapt the framework to use normalized operators for the convenience of studying entanglement measures. The adaptation can be viewed as a mere change of convention.

The process framework first proposed in \cite{oreshkov2012quantum} is an extension of ordinary quantum theory to allow indefinite global causal structure. The main idea is to assume that ordinary quantum theory with fixed causal structure holds locally, while globally the causal structure can be indefinite. Article \cite{oreshkov2016causal} identifies the main assumptions of the framework as local quantum mechanics, noncontextuality, and extendibility. Local quantum mechanics says that ordinary quantum theory with fixed causal structure holds in local parties. Noncontextuality says that the joint outcome probability of local parties is non-contextual, i.e., equivalent local operations lead to the same probabilities. Extendibility says that the operations of local parties can be extended to act on an arbitrary joint ancilla quantum state.

Local parties are where classical outcomes are gathered. We denote the parties by capital letters $A, B, \cdots$, with corresponding input systems $a_1,b_1,\cdots$ and output systems $a_2,b_2,\cdots$ (each system $x$ is associated with a Hilbert space $\mathcal{H}^{x}$). Operations in local parties are quantum instrument with classical outcomes. It is convenient to represent the CP maps of quantum instruments $M:L(\mathcal{H}^{a_1})\rightarrow L(\mathcal{H}^{a_2})$ ($L(\mathcal{H})$ stands for linear operators on Hilbert space $\mathcal{H}$) by their Choi operators \cite{choi1975completely}:
\begin{align}
\mathsf{M}=\abs{a_1a_2} M\otimes \id (\ketbra{\phi_+}{\phi_+})\in L(\mathcal{H}^{a_2}\otimes\mathcal{H}^{a_1}),\label{eq:cs}
\end{align}
where $\id$ is the identity channel on system $a_1$, $\ket{\phi_+}=\sum_i \abs{a_1}^{-1/2}\ket{ii}\in \mathcal{H}^{a_1}\otimes \mathcal{H}^{a_1}$ is a normalized maximally entangled state in a canonical basis on two copies of system $a_1$, and $\abs{x}$ is the dimension of $\mathcal{H}^{x}$. The prefactor $\abs{a_1a_2}$ serves to normalize the process operators to be introduced below.\footnote{One can choose to simultaneously maintain normalization for both the quantum channel Choi operators and the process operators by modifying the composition rule \cite{jia2017quantum}. These are merely different conventions that give the same outcome probabilities.} In our notation, operators are in Sans-serif font, while the corresponding maps are in the normal font.

Under the main assumptions of the framework, a probability assignment to a set of classical outcomes $\{i,j,\cdots, k\}$ in local parties $A, B, \cdots, C$ can be represented through an operator $\mathsf{W}$ as:
\begin{align}\label{eq:pm}
P\big(\mathsf{E}^A_i,\mathsf{E}^B_j,\cdots, \mathsf{E}^C_k\big)=\Tr\{[\mathsf{E}^A_i\otimes\mathsf{E}^B_j\otimes\cdots\otimes \mathsf{E}^C_k]^T \mathsf{W}\}.
\end{align}
Here $\mathsf{E}^X_l$ represents a quantum instrument element inside party $X$ with outcome $l$. $T$ denotes operator transpose. The linear operator $\mathsf{W}\in L(\mathcal{H})$ with $\mathcal{H}:=\mathcal{H}^{a_1}\otimes\mathcal{H}^{a_2}\otimes \mathcal{H}^{b_1}\otimes\mathcal{H}^{b_2}\otimes\cdots \otimes \mathcal{H}^{c_1}\otimes\mathcal{H}^{c_2}$ is called a process operator. The assumptions of the framework imply that
\begin{align}
\mathsf{W}\ge& 0,\label{eq:Wc1}
\\
\Tr\mathsf{W}=&1,\label{eq:Wc2}
\\
\mathsf{W}=&L_V(\mathsf{W}),\label{eq:Wc3}
\end{align} 
where $L_V$ is a projector onto a linear subspace\footnote{The notation is that $_x W:=\omega^x\otimes \Tr_x W$, where $\omega^x$ is the maximally mixed state on system $x$, and that $_{[1-x]}W:=W-_x W$. $\prod_A$ denotes the product over all different parties $A,B,C...$, each with an input and an output system.} \cite{araujo2015witnessing},
\begin{align}
L_V(\mathsf{W}):=_{[1-\prod_A(1-a_2+a_1a_2)+\prod_A a_1 a_2]}\mathsf{W}.\label{eq:Wfc4}
\end{align}
One can understand $\mathsf{W}$ as a generalized density matrix that assigns probabilities to classical outcomes of the quantum instrument elements. It represents a multilinear process map $W$ that map a set of instrument elements to probabilities. We use the word ``process'' to refer to either a process operator or a process map when no ambiguity arises.

We note that ordinary trace-preserving quantum operations can be regarded as special cases of processes, as can be expected since they obey the main assumptions. In particular, both channels and states are processes. Trace non-increasing quantum operations can be regarded as ``subnormalized processes''. The process framework is therefore a generalization of ordinary quantum theory with definite causal structure.

For the ease of expression we also introduce the ``tensorial'' notation following \cite{hardy2011reformulating, hardy2012operator}. Input systems and output systems of local parties are represented as subscripts and superscripts of local operations respectively. For example, $\mathsf{M}$ in (\ref{eq:cs}) can be written as $\mathsf{M_{a_1}^{a_2}}$ to make it clear what the input and output systems are. On the other hand, input systems of local parties correspond to output systems of the processes, and vice versa. Therefore $\mathsf{W}$ in (\ref{eq:pm}) should be written as $\mathsf{W_{a_2b_2\cdots c_2}^{a_1b_1\cdots c_1}}$. We also refer to the same process as $\mathsf{W^{AB\cdots C}}$ using the parties associated with it. We apply the same super- and subscript conventions to the normal font maps and use $\mathsf{W_{a_2b_2\cdots c_2}^{a_1b_1\cdots c_1}}, \mathsf{W^{AB\cdots C}}, W_{a_2b_2\cdots c_2}^{a_1b_1\cdots c_1}$ and $W^{AB\cdots C}$ interchangeably when no ambiguity arises. The process and instrument element maps are linear maps that can be composed (so can their operators). We use repeated super- and subscripts to denote composition, e.g., $N_a^b\rho^a$ denotes the output state on $b$ when the input state $\rho$ on $a$ is fed into the channel $N$ from $a$ to $b$.

\section{LOCC settings}\label{sec:ls}
Bipartite state entanglement measures are essentially defined by the monotonicity axiom, which says that entanglement measures cannot increase under local operation and classical communication (LOCC) \cite{horodecki2009quantum}. Although sometimes not stressed, one needs to specify several ``free parameters'' to determine the LOCC setting. It is important to highlight these parameters before we study entanglement for processes.

The first parameter is the allowed local operations. For states these are quantum channels that map states to states. In special circumstances one might impose restrictions on the channels, e.g., require that the local channels preserve system dimension \cite{horodecki2009quantum}. The second parameter is the allowed direction of classical communication. The allowed classical communication among the parties are categorized into different classes (no classical communication, one-way forward classical communication, one-way backward classical communication, two-way classical communication). The third parameter is the number of rounds of available classical communications. There is a rich structure about the LOCC with different numbers of rounds \cite{chitambar2014everything}.

Processes introduce some new ingredients to both local operations and classical communications. For states the requirement that local operations are channels is justified because they map states to states. However, more general operations map processes into processes. For example, the memory channels\footnote{As shown in \cite{chiribella2009theoretical}, a memory channel of this form is the most general that maps a channel to a channel.} $L_{a_1a_2'}^{a_2a_1'}:=M_{a_1}^{a_1'l}N_{a_2'l}^{a_2}$ takes the process $W^{a_1b_1}_{a_2b_2}$ to a new process $Z^{a_1'b_1}_{a_2'b_2}:=L_{a_1a_2'}^{a_2a_1'}W^{a_1b_1}_{a_2b_2}$ (left of Fig. \ref{fig:lo}). 

\begin{figure}
    \centering
\includegraphics[width=.6\textwidth]{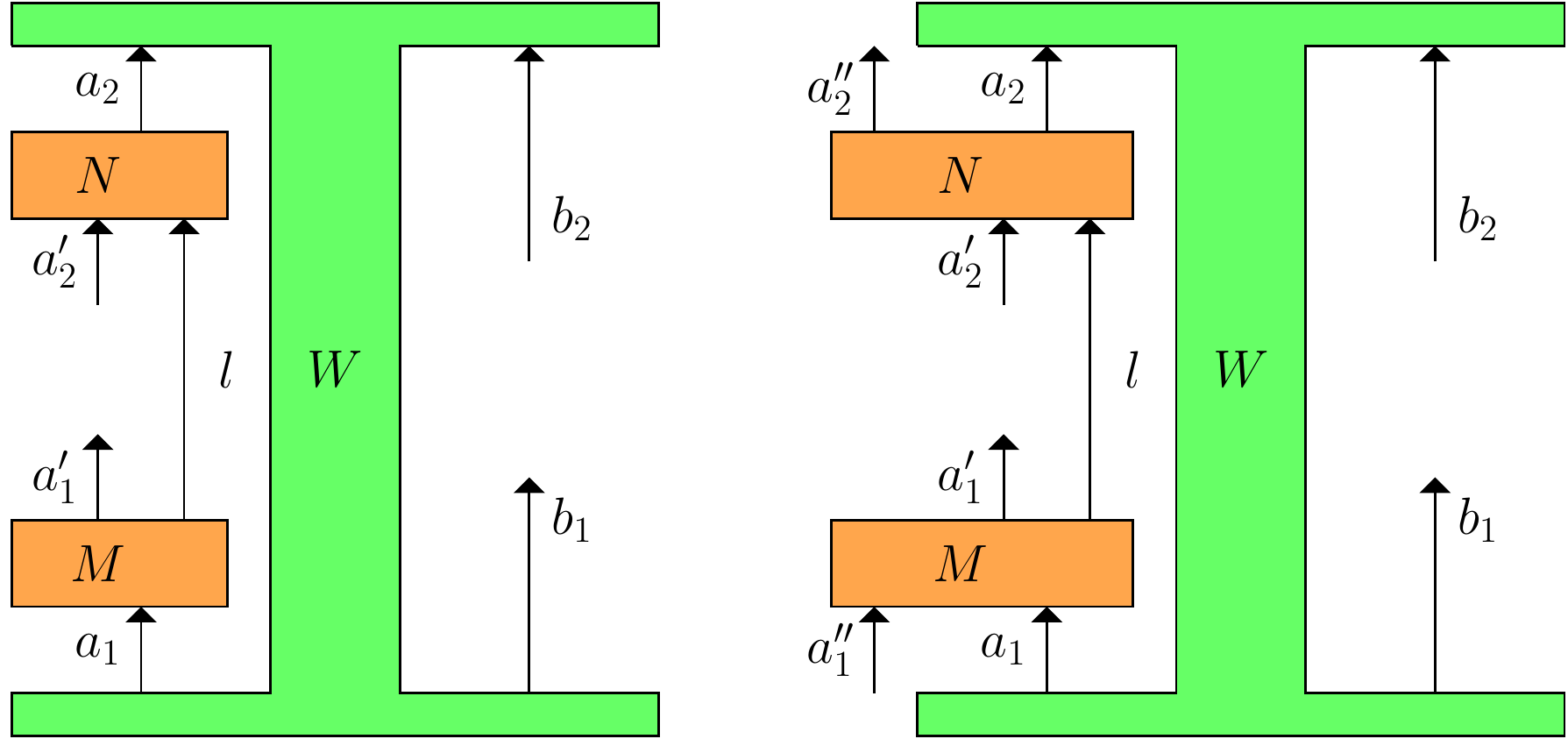}
    \caption{Local operation on processes.}
    \label{fig:lo}
\end{figure}

The extendibility assumption of the process framework allows operations in different local parties to act jointly on a correlated state. This assumption is crucial in setting up the mathematical property of positive semidefiniteness of the process operators, but such unrestricted supply of a correlated state of course cannot be allowed for the entanglement theory. For the local operations of the entanglement theory, we allow for operations whose systems extend beyond those of the processes, but do not supply arbitrary correlated resources. We allow local operations to be memory channels of the form (right of Fig. \ref{fig:lo})
\begin{align}\label{eq:elo}
L_{a_1a_2'a_1''}^{a_2a_1'a_2''}:=M_{a_1a_1''}^{a_1'l}N_{a_2'l}^{a_2a_2''}.
\end{align} This form is useful, for instance, when we use a process $W^{AB}$ to generate a channel $Z^{b_2''}_{a_1''}=L_{a_1a_1''}^{a_2}W^{a_1b_1}_{a_2b_2}K_{b_1}^{b_2b_2''}$. Here $L$ and $K$ fall within the general form of allowed local operations with some systems set trivial. By definition, a process takes a single CP map as input in each local party. Local operations of the form (\ref{eq:elo}) increases the number of systems inside a local party such that more general definite causal structure local operations than CP maps are needed to feed in the systems in order to obtain a probability. Strictly speaking, operations of the form (\ref{eq:elo}) map a process into something else, but we abuse terminology to still call the resulting object a process. The new object is still represented by a trace-one positive semidefinite operator, and the systems inside a local party have a definite causal structure. 


\begin{figure}
    \centering
\includegraphics[width=.7\textwidth]{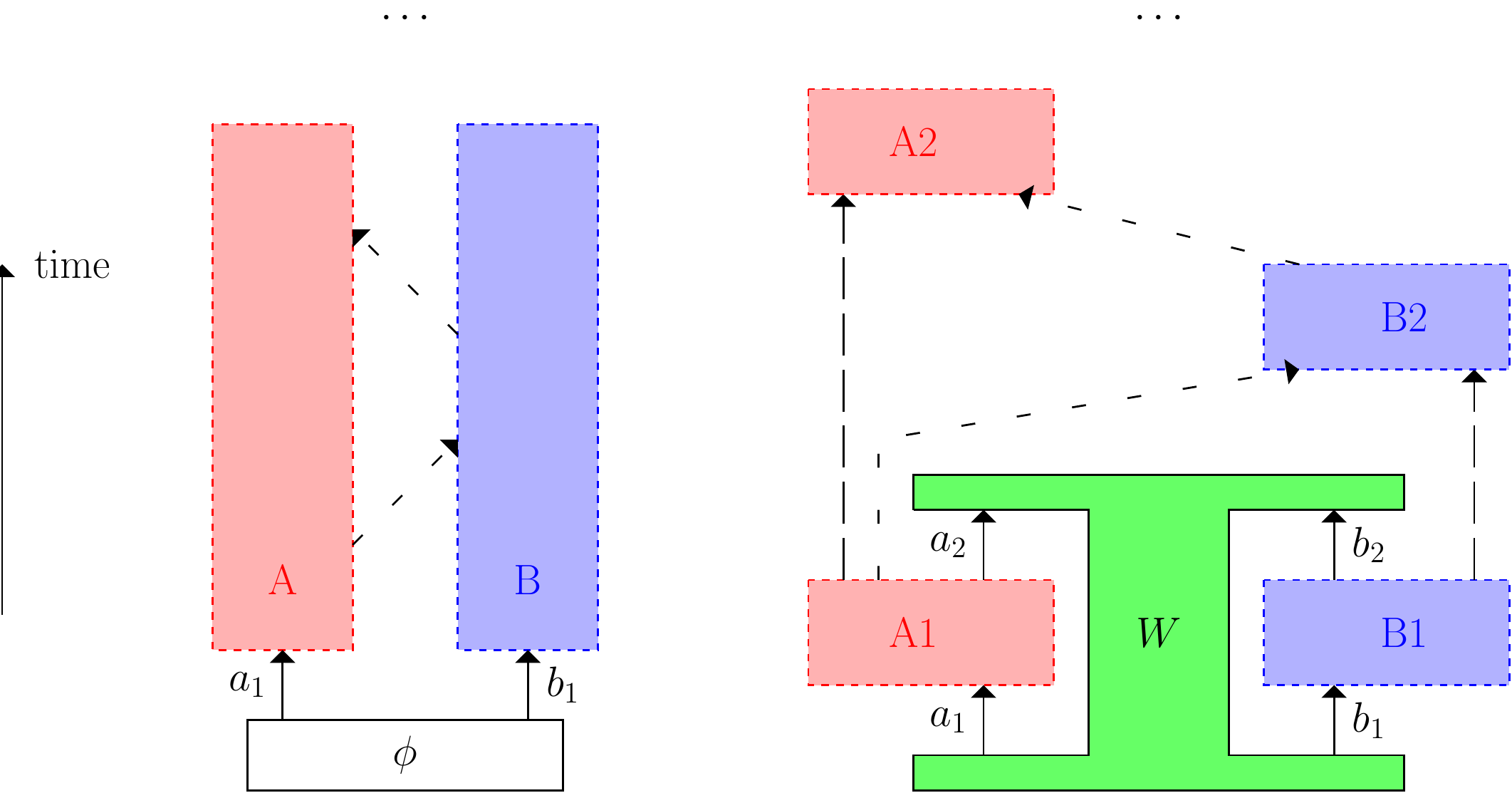}
    \caption{State vs. process LOCC. Densely dashed and loosely dashed arrows represent quantum and classical channels respectively. The state LOCC on the left can be viewed as a special case of the process LOCC on the right, with $W$ reduced to a state $\phi$ when $a_2$ and $b_2$ are taken to be trivial, and all the $A_i$ and $B_i$ respectively coarse-grained into a single $A$ and a single $B$.}
    \label{fig:pslocc}
\end{figure}

In terms of classical communication the new ingredient brought about by processes is that now the local parties have nontrivial causal relations. In the state LOCC paradigm local parties are in general not ``localized'' in time (left of Fig. \ref{fig:pslocc}). If two-way classical communication is allowed the causal relation is trivial since both parties can communicate to each other. To have non-trivial causal structure, local parties must be set to ``localize'' in time (right of Fig. \ref{fig:pslocc}). One can view this as a fine-grained version of the traditional setting such that a party localized in space (such as $A$ in Fig. \ref{fig:pslocc}) is dissected into several copies of the same party in time (such as $A1, A2, \cdots$ in Fig. \ref{fig:pslocc}). Adjacent copies are connected by the quantum identity channel to model time evolution within the local party. Classical communications connect the fine-grained parties and encoding and decoding of classical information are conducted through local operations inside the fine-grained parties.

Assembling all the ingredients together, each particular scheme of allowed local operation and communication forms an \textbf{LOCC setting}. One first specifies the direction and number of rounds of classical communications by introducing copies of local parties in time and equips them with suitable quantum and classical communication channels. One then specifies the allowed local operations inside each copied or original local party. Generally, a LOCC operation changes the number of parties associated with a process. An entanglement measure then needs to be a function defined on all these possible multipartite processes. However, measures defined only on some certain numbers of parties can make sense if the LOCC setting allows one to only create processes with those certain numbers of parties, e.g., the setting may only allow a limited number of rounds of classical communication.

Fig. \ref{fig:glocc} illustrates a general LOCC operation on a process $W$. The copies of $A$ and $B$ are labeled by subscripts from $-m$ to $n$. In ordinary entanglement theory for states the ``past parties'' ($A_0, B_0, A_{-1}, B_{-1}$ etc.) are not considered, because the states only have output systems (which are input systems to the local parties). These past parties are taken into account for processes because they also have input systems. Inside each local party some operations represented by $M_i$ or $N_i$ are already specified as part of the LOCC operation, and some open inputs and outputs are left open for local operations not yet specified to act on. By the assumption of ``local quantum physics'' of the process matrix framework, inside a local party the operations take a definite causal structure. It is shown in \cite{chiribella2009theoretical} that the most general of such operations can be written as ``quantum combs''. In the left figure, the already specified operations are shown as quantum combs with only two teeth due to the limit of space on the page. They actually represent quantum combs with multiple teeth of the form shown in the right figure. As for the classical communications, in the general case they do not have to follow the directions shown in the figure but can take other directions. With this possibility taken into consideration, the figure depicts a most general LOCC operation on the process $W$ (the number of copies of parties can go to infinity by taking $m,n$ or both to infinity). Other special cases can be reproduced by setting some subsystems trivial (one-dimensional). For example, if some classical communication drawn in the figure is absent, this can be realized by setting the input and output dimensions to be $1$ so that no non-trivial classical information can be sent. In total, the local operation applied by $A$ is all the $M_i$ combined, and the local operation applied by $B$ is all the $N_i$ combined. Classical communications are realized through the choosing suitable local operations that interact with the classical channel. The end result of the whole LOCC operation is a process whose subsystems are all the open input and output subsystems. Those that lie within $A_i$ belong to $A$, and those that lie within $B_i$ belong to $B$. If all the local operations are trace preserving (trace-non-increasing), then the final process is still represented by a positive-semidefinite operator with trace 1 (less than or equal to 1). It should be clear from inspecting the figure that apart from $A_1$ and $B_1$, no pair of local parties is in an indefinite causal order. For example, consider $A_0$ and $B_2$. Certainly $A_0$ may signal to $B_2$ through $W$ or the classical channel connecting $A_1$ and $B_2$. However, all output information from $B_2$ only propagates its future parties ($A_3, B_3, A_4, B_4$ etc.), and $A_0$ never receives any information from these parties. Therefore $B_2$ cannot signal to $A_0$. Similar arguments apply to all pairs of parties other than $A_1$ and $B_1$, so there is at most one-way signaling for these parties. LOCC operations do not generate indefinite causal order for parties other than the pair sharing of the original process $W$.

\begin{figure}
\centering
\includegraphics[width=1\textwidth]{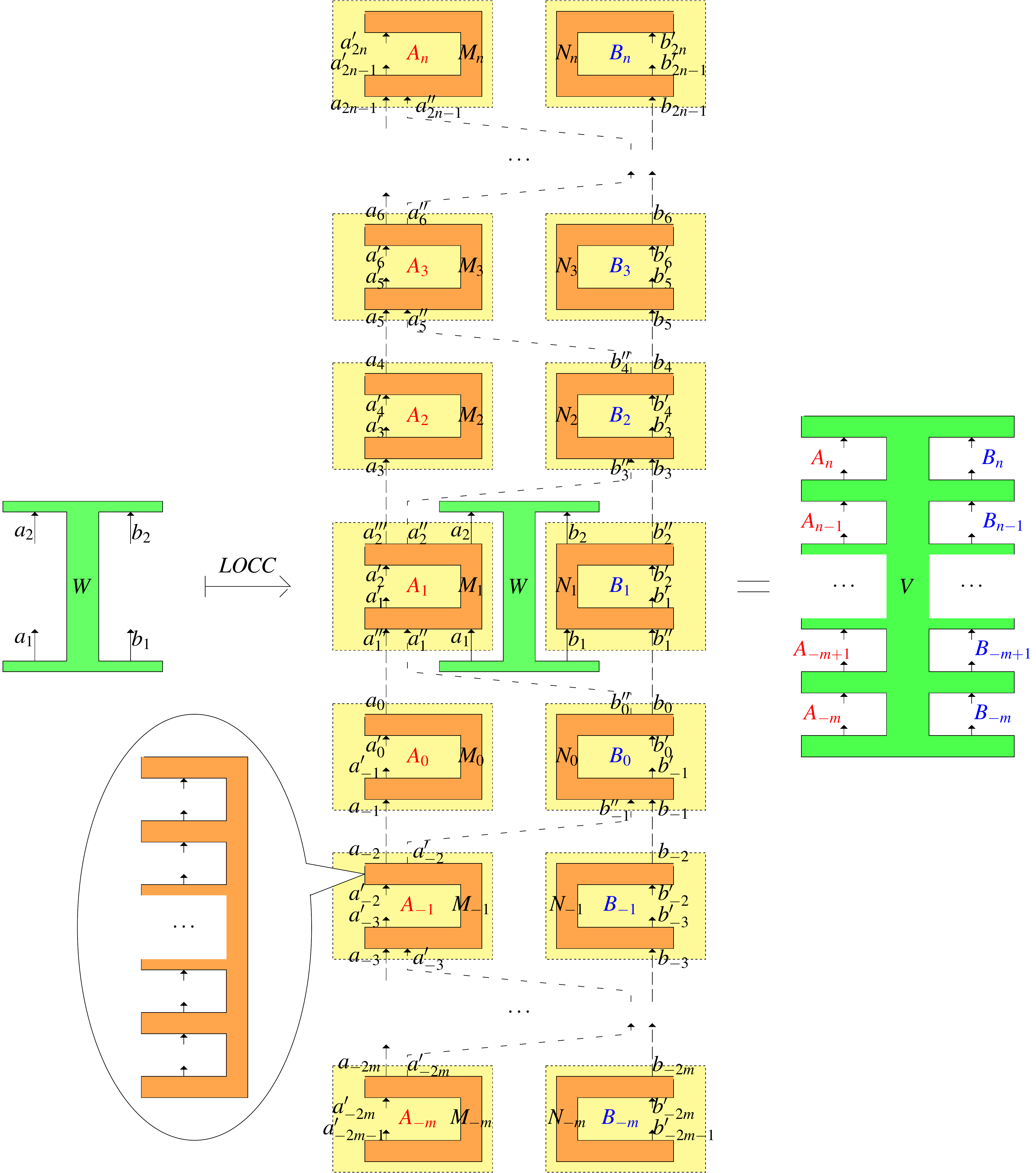}
\caption{A general LOCC operation maps $W$ into $V$. Each local quantum comb $M_i, N_j, -m\le i, j\le n$ is explicitly drawn with two teeth for simplicity, but may contain more teeth as shown in the balloon (correspondingly more systems would be drawn for $V$ on the right). Densely dashed and loosely dashed arrows represent quantum and classical channels respectively.}
\label{fig:glocc}
\end{figure}

In practice the most common two-party (say $A$ and $B$) LOCC settings have no restriction on the local operations or the number of rounds of classical communication. We denote these settings as $LO\rightarrow, LO\leftarrow, LO\leftrightarrow$, which denote LOCC settings with classical communication from $A$ to $B$, from $B$ to $A$, and two-way between $A$ and $B$, respectively.

\section{Process entanglement measures}\label{sec:pem}

Once an LOCC setting is given, the corresponding \textbf{process entanglement measures} can be naturally defined 
as functions on processes that obey the \textbf{monotonicity axiom}, which says that the functions do not increase under LOCC operations. It follows from the monotonicity axiom that a process entanglement measure reaches its minimal value on all processes that can be created through LOCC operations alone. As a convention, we subtract this minimal value from the measure such that LOCC-free processes have zero entanglement. A process whose entanglement measure is positive is then regarded as entangled. 

There is a subtlety about the LOCC settings without classical communications. In LOCC settings with classical communications, all ``separable processes'' of the form $W^{AB}=\sum_i p_i W^A_i$ $\otimes W^B_i$ can be created for free so entanglement measures take the minimal value. In LOCC settings without classical communications only ``product processes'' of the form $W^{AB}=W^A \otimes W^B$ can be created free. If one still wants the entanglement measures to vanish on all separable processes, one needs to impose this as an additional condition. The LOCC setting without classical communication is usually not considered in ordinary entanglement theory for states. In the study of quantum causal structure this case can be important because free classical communication can ``contaminate'' the causal structure by directly establishing causal connections. When one wants to avoid this, an LOCC setting without classical communication is in place.

Restricted to states and meaningful LOCC settings for states (the LOCC operations only result in states), any process entanglement measure reduces to a state entanglement measure by construction. By the above conventions the measure is also always zero on separable states. In this sense process measures generalize state measures.

\section{Coherent-information-based measures}\label{sec:cibm}
Like state entanglement measures, process entanglement measures encompass a family of functions. Within this family, those measures that have operational meaning are of particular interest. In this section we define coherent information for processes and study process entanglement measures based on this definition. The coherent-information-based measures quantify the entanglement generation capacity of processes and serve as examples of ``good'' process entanglement measures that have clear operational meaning.

\subsection{Coherent information for processes}\label{sec:cip}
Ordinarily coherent information is defined for states and channels \cite{wilde2017quantum}. We generalize the definition to processes. For a bipartite state $\rho^{ab}$, the coherent information with target system $b$ is defined as
\begin{align}\label{eq:cis}
I^b(\rho^{ab})&:=S^b-S^{ab}.
\end{align}
$S^x$ is the von Neumann entropy of the (reduced) state on system $x$. $I^b(\rho^{ab})$ is minus the conditional entropy $S^{a|b}(\rho^{ab}):=S^{ab}-S^b$, so it can be positive, zero, or negative. For a pure state, the coherent information coincides with the entanglement entropy. Intuitively, the coherent information measures the quantum correlation contained in a bipartite state.

The coherent information for a channel $N_a^b$ with target system $b$ is defined as
\begin{align}
I^b(N_a^b):=\sup_{\rho} I^b(N_a^b \rho^{aa'}),
\end{align}
where the supremum is over input states $\rho^{aa'}$ with arbitrary auxiliary system $a'$.\footnote{A more commonly seen but equivalent \cite{jia2017quantum} definition is to optimize only over pure states.} By the data-processing inequality for coherent information \cite{schumacher1996quantum, wilde2017quantum}, local operations $M_b^{b'}$ on $b$ do not increase the coherent information of the final state. Therefore one can equivalently define the channel coherent information as
\begin{align}\label{eq:cic}
I^{b}_{LO}(N_a^b):=\sup_{\rho,M} I^{b'}(M_b^{b'}N_a^b\rho^{aa'}).
\end{align}
The subscript $LO$ stands for optimization over local operations represented by $\rho$ and $M$ at the input and output ends of the channel respectively. This definition suggests the operational interpretation of $I^b(N_a^b)$ as a measure of the maximal state coherent information that can be established through the channel when the input and output parties are free to choose $\rho$ and $M$.

This interpretation motivates us to define \textbf{process coherent information} as follows.\footnote{Incidentally, other entropic measures such as mutual information can be generalized to apply to processes analogously. For example, analogous to (\ref{eq:defci}) process mutual information can be defined as $I_R^{x:y}(W^{xy}):=\sup_{O_R} I^{x:y}(O_R(W^{xy}))$, where $I^{x:y}(\rho^{xy}):=S^x+S^y-S^{xy}$.} The coherent information of the process $W$ with target systems $x$ and supplemental resource $R$ is defined as
\begin{align}\label{eq:defci}
I_R^x(W):=\sup_{O_R} I^x(O_R(W)).
\end{align}
$R$ labels the supplemental operations $O_R$ allowed for optimization (e.g. local operations and classical communications). $O_R$ maps the original process $W$ to a new process $O_R(W)$ whose systems include the target systems $x$ as a subset. Through the Choi isomorphism $O_R(W)$ is represented as a positive semi-definite operator such that its coherent information can be calculated using (\ref{eq:cis}). Often we are interested in taking some parties $X\subset \{A,B,\cdots\}$ as the target and in this case $x$ is all the systems that belong to the parties of $X$. In this case we can write $I_R^X(W)$ for $I_R^x(W)$.

As an example of process coherent information, the channel coherent information (\ref{eq:cic}) is a special case when $W=N$, $X=B$, and $R=LO$.

\subsection{Entanglement measures}
In this section we consider process entanglement measures built out of process coherent information. Consider the quantity
\begin{align}
I^x_{LOCC}(W)=\sup_{O_{LOCC}}I^x(O_{LOCC}(W)),
\end{align}
where the optimization is over operations $O_{LOCC}$ allowed in some LOCC setting. Under this LOCC setting, $I^x_{LOCC}(W)$ is monotonic by construction, and hence is a process entanglement measure.

In the asymptotic setting, it is useful to consider the regularized coherent information
\begin{align}
\mathfrak{I}_{R}^x(W):=\lim_{k\rightarrow \infty}\frac{1}{k} I^x_R(W^{\otimes k}).   
\end{align}
In particular, $\mathfrak{I}_{LOCC}^x(W)$ obeys the monotonicity axiom, and is hence is also a process entanglement measure. We note that not all processes $W$ have $W^{\otimes k}$ as valid processes, and a necessary and sufficient condition for $W$ to have valid self-tensor product is given in \cite{jia2017process}. We call those processes that allow self-tensor products ``self-productible processes''. The regularized coherent information can be defined for them without further specifications. Another way to deal with the subtlety of tensor products is to restrict the allowed local operations. With suitable restrictions, self-tensor products can be defined for all processes. A particular physically motivated construction is the sequential setting where the processes appear in a definite causal order (the parties within an individual process do not need to) \cite{jiageaneralizing}. The regularized coherent information can be defined for all processes in such settings.

\subsection{Process entanglement generation capacity}
The following theorem is a direct generalization of the result of \cite{horodecki2000unified} to processes. It says that the process entanglement measures $\mathfrak{I}^B_{LO\rightarrow}(W^{AB})$ and $\mathfrak{I}^B_{LO\leftrightarrow}(W^{AB})$ have the operational meaning of classical-communication-assisted entanglement generation capacities. The theorem yields quantum communication and entanglement distillation capacities, respectively, for channels and states as special cases.

\begin{theorem}\label{th:pegc}
The classical-communication-assisted entanglement generation capacity for a self-productible two-party process $W^{AB}$ is given by $\mathfrak{I}^B_{R}(W^{AB})$, where $R\in \{LO\rightarrow,LO\leftrightarrow\}$ is the allowed assistance.
\end{theorem}
\begin{proof}[Sketch of proof]
We first show achievability. The protocol is to use some LOCC operation to turn $W^{\otimes n}$ into a shared state, and then perform entanglement distillation. The optimal rate of entanglement generation $C_{R}(W)$ under $R$ using $W$ satisfies
\begin{align}
C_{R}(W^{\otimes n})\ge & D_{R}(O_{R}(W^{\otimes n}))
\\
\ge & I^B(O_{R}(W^{\otimes n})).
\end{align}
$O_R$ is any LOCC operation allowed by $R$, and $D_R(\rho)$ is the distillable entanglement of $\rho$ under $R$. 

One first applies some local operations $O_R$ on $W$. $O_{R}(W^{\otimes n})$ may not be a state, but feeding in a maximally entangled state to each open input system turns it into a state $\rho$ represented by the same operator. Without loss of generality we assume this is carried out. Then one can perform entanglement distillation on the state $O_{R}(W^{\otimes n})$ to obtain $D_{R}(O_{R}(W^{\otimes n}))$ entanglement, which justifies the first line. The second line follows from the Hashing inequality proved in \cite{devetak2005distillation}, which says that $D_{LO\rightarrow}(\rho)\ge I^B(\rho)$, and implies that $D_{LO\leftrightarrow}(\rho)\ge I^B(\rho)$. Optimization over $O_R$, dividing both sides by $n$ and taking the $n\rightarrow \infty$ limit, one gets that the rate $\mathfrak{I}^B_{R}(W^{AB})$ is achievable. 

We now move to the converse part of the proof. It holds that
\begin{align}
C:=&C_{R}(W^{\otimes n})=I^B(\rho_+^{\otimes C})
\\ \le& I^B(E_{R}(W^{\otimes n}))+\epsilon_n,
\end{align}
where $E_R$ is an optimal operation for entanglement generation and $\epsilon_n\rightarrow 0$ in the large $n$ limit. The first line holds because each singlet $\rho_+$ has coherent information $1$. The second line holds because the optimal operation generates a state $E_{R}(W^{\otimes n})$ that is close to $C$ singlet and because coherent information is continuous (AFW inequality \cite{alicki2004continuity, winter2016tight, wilde2017quantum}). Dividing both sides by $n$ and taking the $n\rightarrow \infty$ limit, one gets that $\mathfrak{I}^B_{R}(W^{AB})$ upper-bounds the capacity. 
\end{proof}
The theorem is adaptable to general processes under suitable restrictions on local operations. For example, the theorem can be adapted to hold for the entanglement generation task in the sequential setting mentioned above below the definition of regularized coherent information \cite{jiageneralizing}.

The entanglement generation capacity for processes is a process entanglement measure by definition. It can be viewed as a generalization of the distillable entanglement for states. The entanglement generation capacity for processes reduce to the quantum communication capacity for channels \cite{barnum2000quantum}. From this perspective, quantum channel capacities $Q_{R}$ are in fact entanglement measures. Specifically, we have $Q_{LO}=\mathfrak{I}^B_{LO}, Q_{LO\rightarrow}=\mathfrak{I}^B_{LO\rightarrow}, Q_{LO\leftarrow}=\mathfrak{I}^A_{LO\leftarrow}$ and $Q_{LO\leftrightarrow}=\mathfrak{I}^B_{LO\leftrightarrow}=\mathfrak{I}^A_{LO\leftrightarrow}$ \cite{horodecki2000unified}. Forward classical communication does not increase the quantum channel capacity \cite{barnum2000quantum}, so $Q_{LO}$ and $Q_{LO\rightarrow}$ are entanglement measures under the $LO\rightarrow$ setting. $Q_{LO\leftarrow}$ and $Q_{LO\leftrightarrow}$ are entanglement measures under the $LO\leftarrow$ and $LO\leftrightarrow$ settings, respectively.

\section{Entanglement for channels and networks}\label{sec:ecn}
As we saw, the generalization of entanglement to processes opens up the possibilities to study entanglement of channels and networks as special cases of processes. Viewed from a spacetime perspective we can now generalize consideration of entanglement from bifurcation of spacelike regions to bifurcation of spacetime regions. In this section we illustrate these possibilities through toy model examples. For concreteness we focus on coherent-information-based entanglement measures.

\subsection{Channel}
Consider a quantum channel with two input systems $a, b$ and two output systems $c, d$. Fig. \ref{fig:ste1} illustrates different ways to bifurcate the channel --- let one party $A$ have access to the shaded systems and another party $B$ have access to the rest. The goal is to calculate coherent-information-based entanglement measures.

Unitary channels are of particular interest among all channels. 
The Choi state of a unitary is a maximally entangled state so its von Neumann entropy vanishes. Consequently,
\begin{align}
I^A(N_{ab}^{cd})=S^{A}-S^{abcd}=S^A=S^{B}=S^{B}-S^{abcd}=I^B(N_{ab}^{cd}).
\end{align}
Here $S^x$ denotes the von Neumann entropy of the (reduced) Choi operator of the channel $N$. The equalities hold because for a pure state $\rho^{AB}$, $S^A=S^B$. This says that it does not make a difference which party we set as the target --- without loss of generality let it be $A$. Moreover, partially tracing out one system of a maximally entangled state yields the maximally mixed state, so
\begin{align}
S^{ab}=S^{cd}=\log \abs{ab}=\log \abs{cd},
\end{align}
where $\abs{x}:=\dim \mathcal{H}^{x}$ ($N$ is unitary so $\abs{ab}=\abs{cd}$), and $S^x=\log \abs{x}$ for $x=a,b,c,d$.

\begin{figure}
    \centering
\includegraphics[width=.6\textwidth]{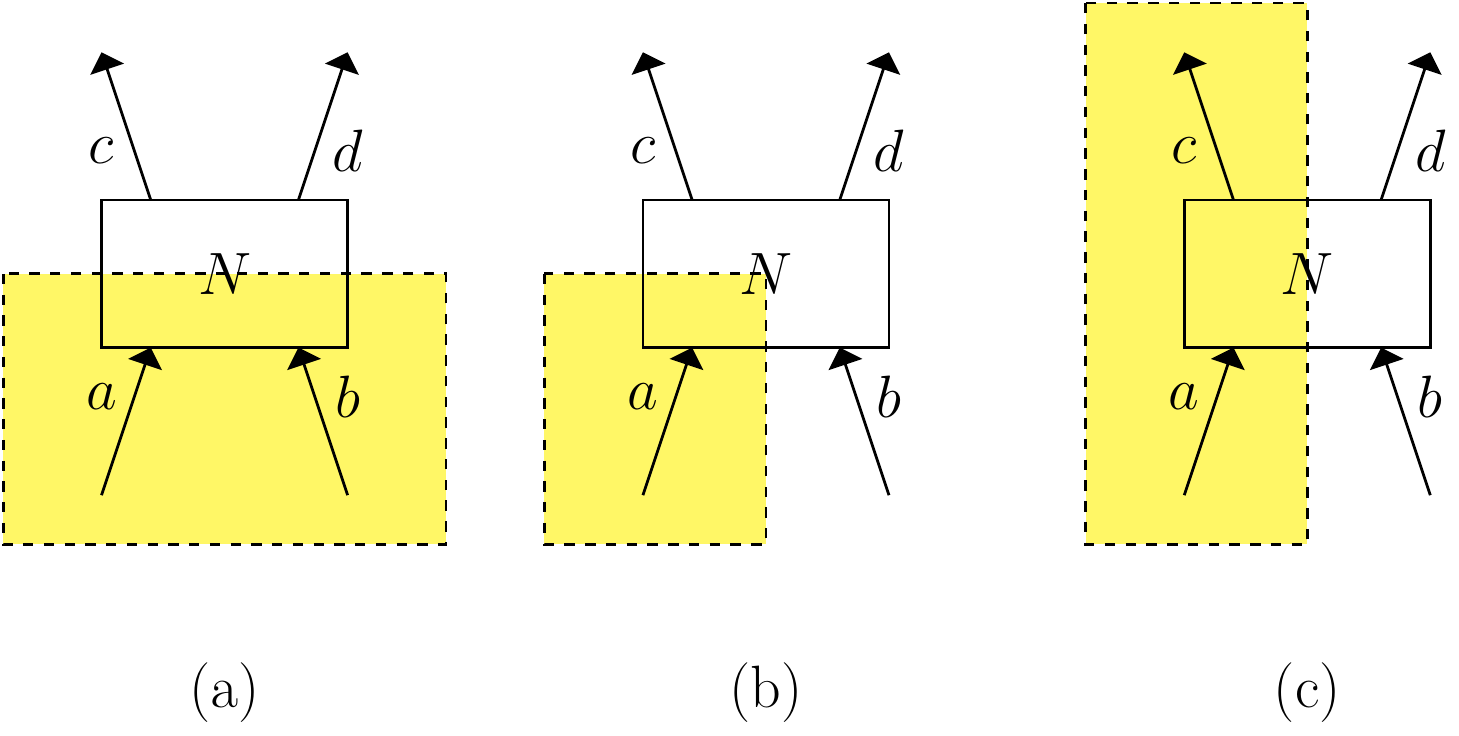}    \caption{Entanglements of a channel}
    \label{fig:ste1}
\end{figure}

Taking these into account, we have for cases (a) and (b) in Fig. \ref{fig:ste1}, respectively, 
\begin{align}
I^A(N_{ab}^{cd})=&S^{ab}-S^{abcd}=S^{ab}=\log\abs{ab},
\\
I^A(N_{ab}^{cd})=&S^{a}-S^{abcd}=S^{a}=\log\abs{a}.
\end{align}
Case (c) is a bit more complicated. We have
\begin{align}
I^A(N_{ab}^{cd})=&S^{ac}-S^{abcd}=S^{ac}.
\end{align}
This value is undetermined. For instance, if $N_{ab}^{cd}=H_{a}^{c}\otimes G_{b}^{d}$ factors into two unitaries $H$ and $G$, then $S^{ac}=0$. If instead $N_{ab}^{cd}=H_{a}^{d}\otimes G_{b}^{c}$ factors into two unitaries $H$ and $G$ in another way, then $S^{ac}=S(\omega^{ac})=\log\abs{ac}$. Nevertheless, if $N_{ab}^{cd}$ is random and the bifurcation results in two systems of equal size, i.e., $\abs{ac}=\abs{bd}$, then one can show that \cite{hosur2016chaos}
\begin{align}
I^A(N_{ab}^{cd})> \log \abs{ac}-1.
\end{align}
Under these assumptions, in all three cases of Fig. \ref{fig:ste1} the entanglement equals or approximately equals the number of qubits in the target system. This implies that optimizing over local operation and classical communication, or over regularization, can only have zero or negligible increase for the coherent information, so we conclude that the entanglement is maximal or almost maximal for these cases of ``timelike'', ``spacetime'' or ``spacelike'' bifurcations of the unitary channels.

\subsection{Network}
\begin{figure}
    \centering
\includegraphics[width=.66\textwidth]{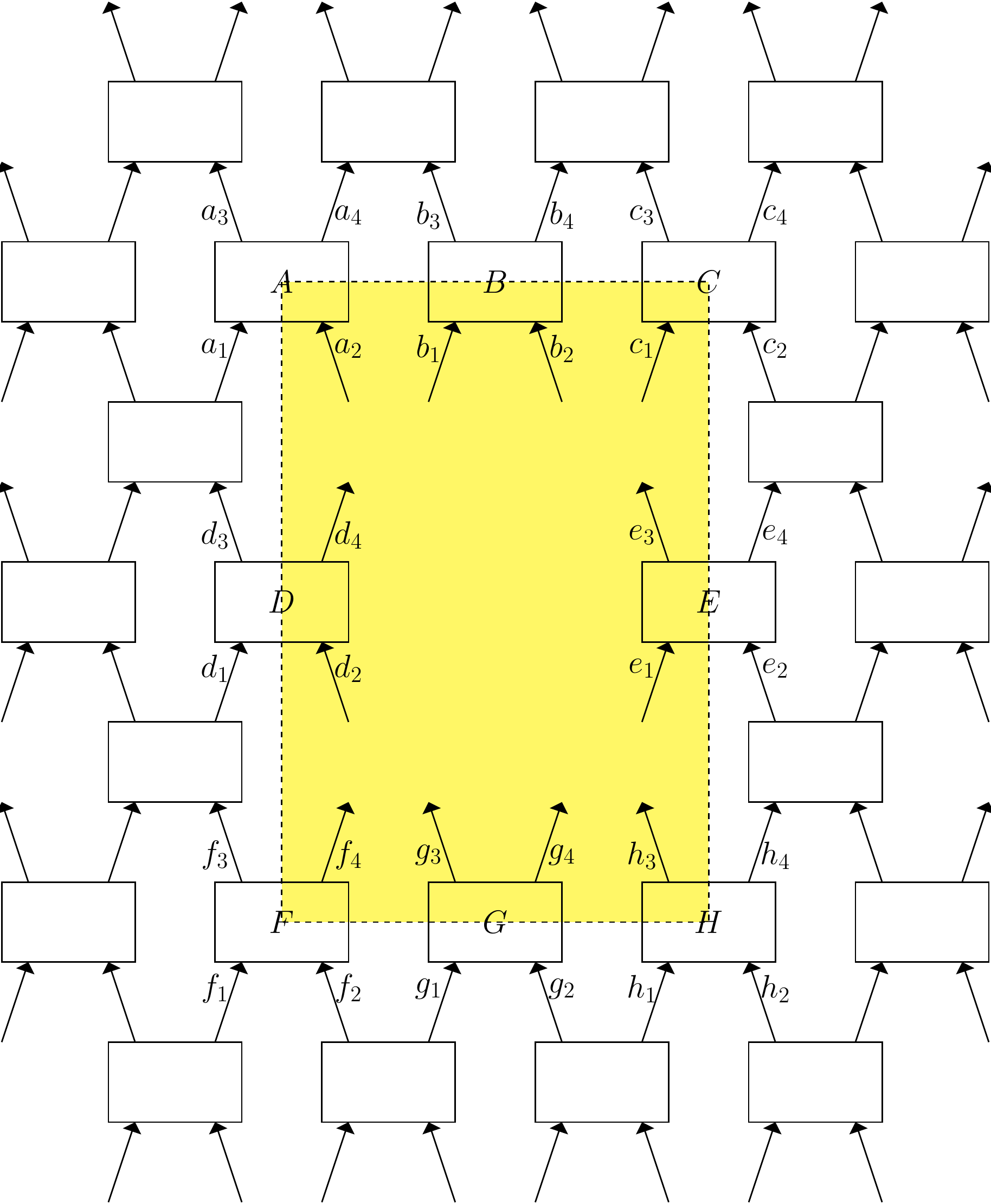}    \caption{Bipartition of a quantum network.}
    \label{fig:ste2}
\end{figure}

To see how entanglement can be considered for quantum networks and bipartition of spacetime regions, consider the example in Fig. \ref{fig:ste2}. This quantum network is composed of many unitary channels drawn as boxes, and can be viewed as a toy model whose continuum limit is a quantum field theory \cite{oreshkov2016operational, hardy2016operational}. It is assumed that the network is foliated according a global time such that each horizontal set of channels act at the same time. We break the network into two parts of the shaded and unshaded regions and ask for the coherent-information-based entanglement measures with the shaded region as the target.

Since every box is unitary, the whole network is represented by a pure Choi operator. To calculate the coherent information of the bipartition, we only need to trace out all the open subsystems of the unshaded region and calculate the entropy of the reduced network. Due to the global time foliation, information in the output subsystems $b_3b_4$ of $B$ must flow to the global future in the end and be traced out. Therefore the reduced network on subsystem $b_1b_2$ is represented by the maximally mixed state $\omega^{b_1b_2}$, like in case (a) of the last section. Likewise on $g_3g_4$ we have $\omega^{g_3g_4}$. Similarly, information in $a_3a_4$ of $A$ must flow to the global future and be traced out to yield a maximally mixed state, which implies that the reduced operator on $a_2$ is the maximally mixed state, like in case (b) of the last section. The same conclusion holds for $c_1, f_4$, and $h_3$.

Information in $d_3$ must also flow to the global future as there is no way it gets to the shaded part. Similarly, information in $d_1$ must originate in the global past. Hence like in case (c) of the last section, if $D$ is a random unitary and $\abs{d_2d_4}=\abs{d_1d_3}$, the reduced operator on $d_2d_4$ is approximately maximally mixed. The same conclusion holds for $e_1e_3$ under the same assumptions.

Altogether, if $D$ and $E$ are random, the coherent information with the shaded region as target is approximately maximal. In the limit of large system size the difference from maximal coherent information is negligible (optimization over LOCC or regularization becomes redundant), and we conclude that in this case the entanglement measured by coherent-information based measures $I^B_{LOCC}$ or $\mathfrak{I}^B_{LOCC}$ is maximal.


\section{Outlooks}\label{sec:out}
In this paper, we showed how the notions of entanglement and entanglement measures can be generalized to apply to all states, channels, networks, and processes. One of the main motivations for the generalization is to enable the study of entanglement in the presence of indefinite causal structure in quantum gravity. It had been suggested that with the generalized entanglement, indefinite causal structure may regularize divergences in quantum field theory, elucidate black hole information processing problems, and explain ``dark energy'' \cite{jia2017quantum}. These suggestions are only preliminary and need further development, and a lot of work needs to be done to check the implication of the generalized entanglement on fundamental problems in quantum gravity. 

Apart from applications to particular questions, there are several directions to develop the theory of generalized entanglement itself. When can one turn a state entanglement measure into a process measure? Can one generalize multipartite entanglement along similar lines of reasoning? Can one allow even more general local operations such as processes with indefinite causal structure? Can one generalize the study to general probabilistic theories and time-symmetric theories?

The particular entanglement measures we studied in some detail --s coherent-information-based-measures, contain an optimization over LOCC. Such an optimization is in general hard to compute exactly. However, it is possible to find efficiently computable bounds for the quantities. As mentioned, in a special case where the process is a state the measure reduces to the distillable entanglement of the state. No general method to compute this quantity is known, but there are computable bounds such as logarithmic negativity \cite{vidal2002computable, plenio2005logarithmic} and the semidefinite programming bounds in \cite{rains2001semidefinite} and \cite{wang2016improved}. Because distillable entanglement is a special case of coherent-information-based entanglement measures and also of process entanglement generation capacity, there are potential ways to generalize these bounds to apply to processes.

Finally, the upgrade of the LOCC paradigm to account for causal structure in a fine-grained way suggests that the paradigm itself may be shifted. The fact that classical communications can be viewed as part of the process suggests a different way to look at the LOCC paradigm. Instead of taking the initial resource (such as a state) as the starting point and adjoining classical communications to it afterwards, one can take the initial resource and the classical communications together as the given resource. Then the LOCC paradigm may be upgraded into an ``LO paradigm'', where only local operations are free. In this LO paradigm, whether some classical communication is allowed is determined by the given resource itself. An entanglement measure is required to obey the monotonicity under local operations axiom. This new LO paradigm is just an equivalent formulation of the LOCC paradigm, but it allows one to maintain that all communications among the local parties are mediated by the given resource itself. It is yet unclear what advantage (or disadvantage) this new perspective offers in practical applications, but at least for intuition gaining it accounts for causal structure of local parties in a nontrivial and more natural way.


\section*{Acknowledgements}
The author thanks Lucien Hardy and Achim Kempf for fruitful guidance and generous support as supervisors, as well as Jason Pye, Robert Oeckl and Tian Zhang for helpful discussions. The author also thanks the organizers Natacha Altamirano, Fil Simovic, and Alexander Smith, the host Robert Mann, and all the participants of the Spacetime and Information Workshop at Manitoulin Island for offering a chance to improve the work. A special thanks goes to Fabio Costa for many valuable discussions, particularly about the form of the allowed local operations.

\bibliographystyle{apsrev}
\bibliography{bib.bib}

\end{document}